\documentclass[11pt] {amsart}

\usepackage{amsthm}
\usepackage{amsxtra}
\usepackage{amssymb}
\usepackage{graphicx}

\usepackage{color}
\usepackage{amsfonts}
\usepackage{textcomp}
\usepackage{lmodern}

\newtheorem{theorem}{Theorem}

\newtheorem{defin}{Definition}
\newtheorem{lemma}{Lemma}

\newcommand{\bl}{\begin{flushleft}}
\newcommand{\el}{\end{flushleft}}
\newcommand{\br}{\begin{flushright}}
\newcommand{\ert}{\end{flushright}}
\newcommand{\bc}{\begin{center}}
\newcommand{\ec}{\end{center}}
\newcommand{\mcal}[1]{\mathcal{#1}}

\newcommand{\recip}[1]{\frac{1}{#1}}
\newcommand{\imply}{\Rightarrow}

\newcommand{\numList}{\begin{enumerate}}
\newcommand{\enumList}{\end{enumerate}}

\newcommand{\composed}{\text{\textopenbullet}}

\newcommand{\e}{\epsilon}

\newcommand{\re}{\mathbb{R}}

\newcommand{\nn}{\nonumber\\}
\newcommand{\la}{\langle}
\newcommand{\ra}{\rangle}

\newcommand{\fint}{-\!\!\!\!\!\!\int}

\numberwithin{equation}{section}

\title [Quantum Ergodicity for a Class of Mixed Systems]{Quantum Ergodicity for a Class of Mixed Systems}
\author[J. Galkowski]{Jeffrey Galkowski}
\address{Mathematics Department, University of California, Berkeley, 
CA 94720, USA}
\email{jeffrey.galkowski@math.berkeley.edu}

\begin{document}

\begin{abstract}
We examine high energy eigenfunctions for the Dirichlet Laplacian on domains where the billiard flow exhibits mixed dynamical behavior. (More generally, we consider semiclassical Schr\"{o}dinger operators with mixed assumptions on the Hamiltonian flow.) Specifically, we assume that the billiard flow has an invariant ergodic component, $U$, and study defect measures, $\mu$, of positive density subsequences of eigenfunctions (and, more generally, of almost orthogonal quasimodes). We show that any defect measure associated to such a subsequence satisfies $\mu|_{U}=c\mu_L|_{U}$, where $\mu_L$ is the Liouville measure. This proves part of a conjecture of Percival \cite{perc}. 
\end{abstract}

\maketitle

\section{Introduction}
The distribution of eigenfunctions over phase space in the semiclassical limit is an important object of study in the theory of quantum chaos. The fundamental result is a quantum ergodicity theorem of Shnirelman \cite{schnir}, Zelditch \cite{zeld}, and Colin de Verdi$\grave{\text{e}}$re \cite{dever} which states that, for classically ergodic systems, high energy eigenfunctions distribute uniformly in phase space. Since we will study domains with boundary, we note that G\'{e}rard-Leichtnam \cite{Gerr} and Zelditch-Zworski \cite{ZeZw} generalized quantum ergodicity to that case. 

Some progress has been made toward understanding semiclassical limits of eigenfunctions for systems with dynamical behavior that is not completely ergodic. Marklof and O'Keefe, \cite{QESeparated}, examine separated phase spaces for certain maps. More recently, Marklof and Rudnick \cite{Rudnick} have made further strides towards an understanding of quantum behavior for mixed phase space. In particular, they prove that, for rational polygons, the eigenfunctions of the Dirichlet Laplacian equidistribute in configuration space. We use \cite{Rudnick} as inspiration for further work on semiclassical limits for systems with mixed dynamical behavior. 

We examine systems whose phase spaces have invariant subsets, $U$, such that the flow restricted to $U$ is ergodic and the Liouville measure of $U$ is positive. In particular, let $(M,g)$ be a smooth Riemannian manifold of dimension $d$ with a piecewise smooth boundary, $\partial M$. Then $M\subset \widetilde{M}$, where $\widetilde{M}$ is a manifold without boundary to which $g$ extends smoothly, and $\partial M=\bigcup_{j=1}^JN_j$ where $N_j$ are smooth embedded hypersurfaces in $\widetilde{M}.$ Define $\partial^oM\subset\partial M$ to be the open set of all points where the boundary is smooth. Then the complement $\partial M\setminus \partial^oM$ has measure zero.

\begin{figure}[htbp]
\includegraphics[width=5.5in]{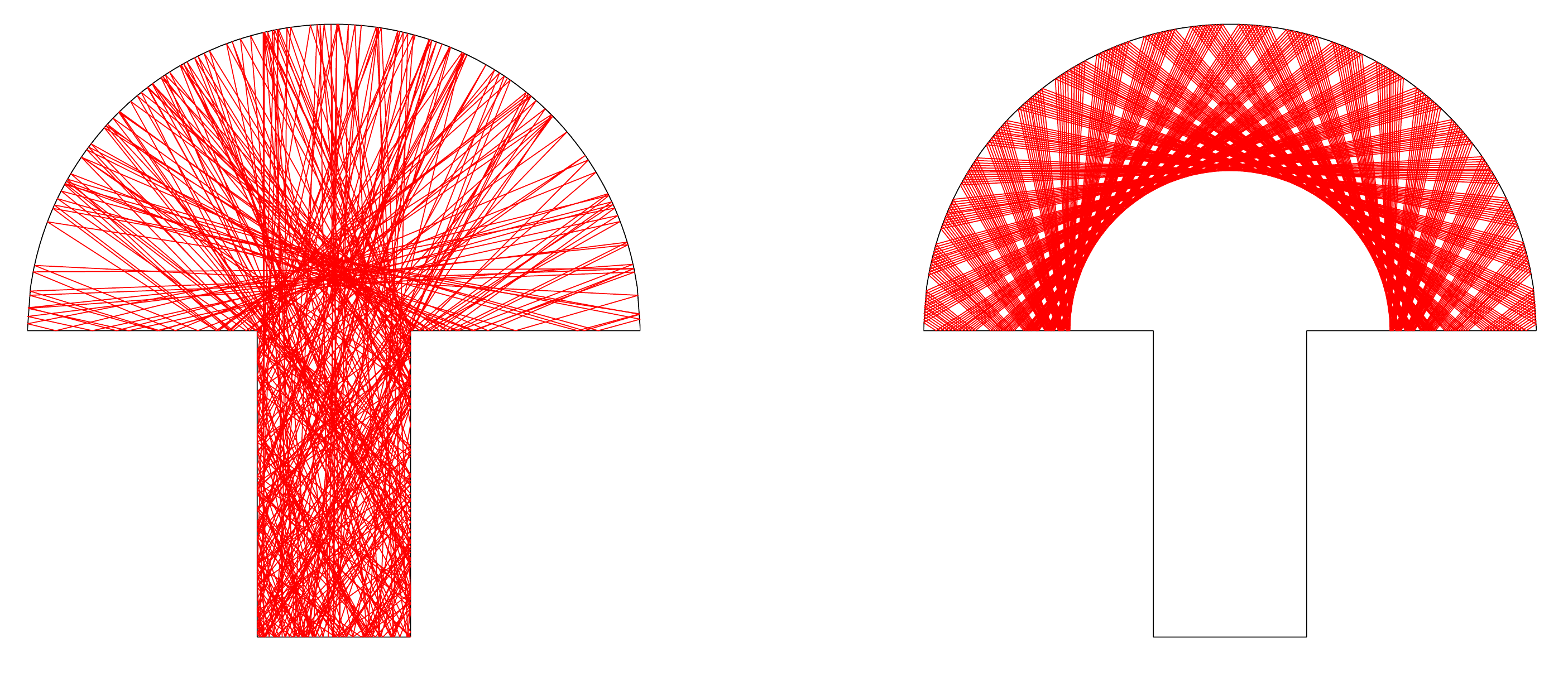}
\begin{center}
\caption{\label{f:1}
We show two billiards trajectories in a Bunimovich mushroom \cite{Bunim} with semicircular hat of radius 1 and centered base of width 1/2 and height 1. On the left, we have a trajectory in the ergodic portion of phase space. On the right, we have one in the integrable portion.}
\end{center}
\end{figure}
\begin{figure}[htbp]
\includegraphics[width=5.25in]{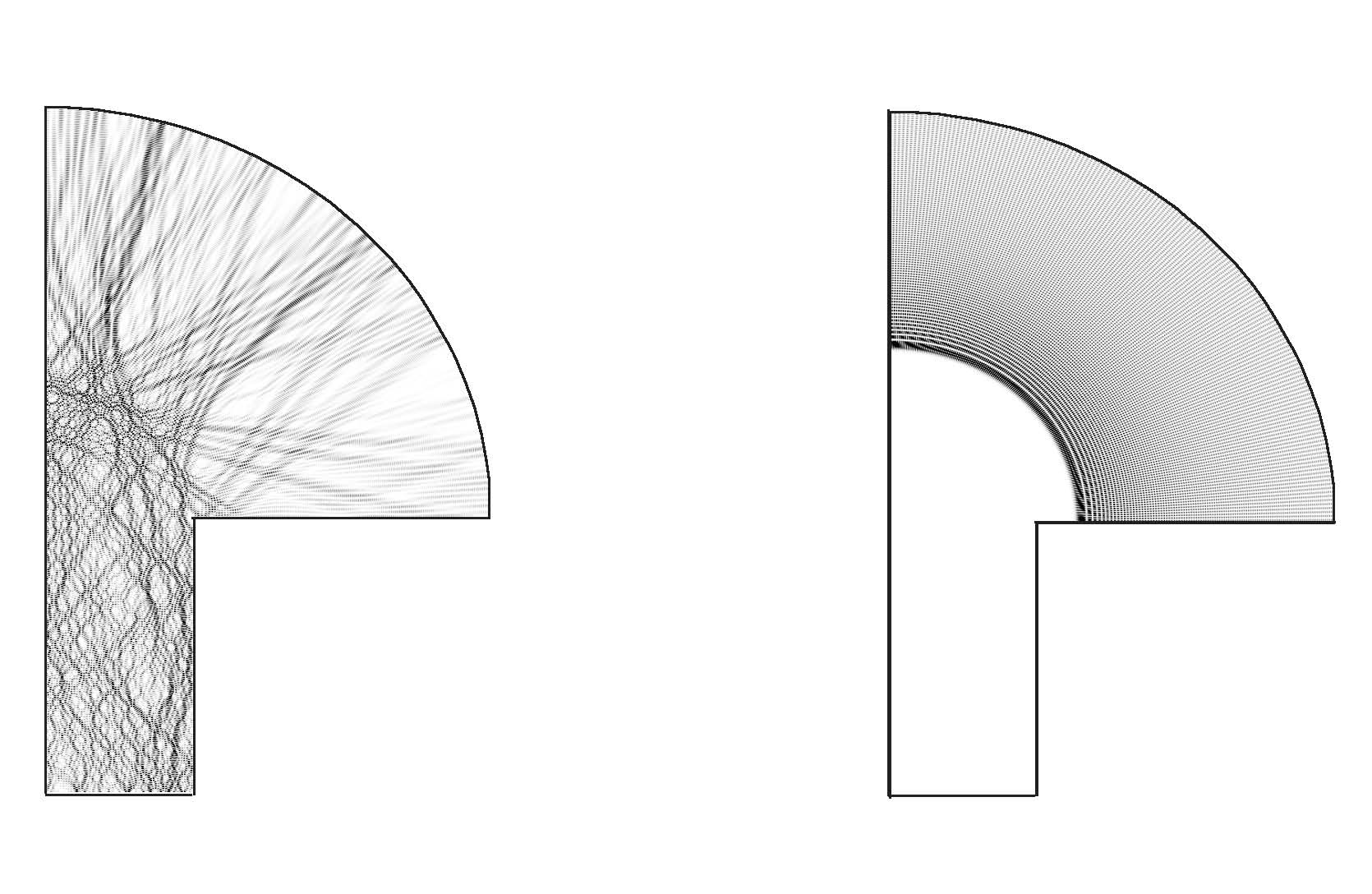}
\begin{center}
\caption{\label{f:2}
We show two high energy eigenfunctions for the laplacian in a Bunimovich mushroom \cite{Bunim}. The left hand eigenfunction spreads uniformly through the ergodic protion of phase space, and the right hand eigenfunction concentrates in the integrable portion. (Images courtesy of A. Barnett and T. Betcke \cite{barn}).}
\end{center}
\end{figure}
We consider $P(h)=-h^2\Delta +V$ with Dirichlet boundary conditions and let $p(x,\xi)$ be its principal symbol. Since $p(x,\xi)$ is smooth up to the boundary, we can extend it smoothly to $T^*\widetilde{M}$. We assume that
\begin{equation}
\label{eqn:assumeE}
\text{for }E\in [a,b],\text{  }dp|_{p^{-1}(E)}\neq 0
\end{equation}

and 
\begin{equation}
\label{eqn:PassumeV}
x\in \partial^o M, V(x)=E\imply dV\notin N_x^*\partial^oM.
\end{equation}
so that $p^{-1}(E)$ and $T_{\partial^oM}^*M$ intersect transversally. 

We then write 
$$p^{-1}(E)\cap T^*_{\partial^oM}M=\Omega_E^+\sqcup \Omega_E^-\sqcup\Omega_E^0,$$
where $(x,\xi)$ lies in $\Omega_E^+$ if the vector $H_px\in T\widetilde{M}$ points outside of $M$, in $\Omega_E^-$ if it points inside $M$, and $\Omega_E^0$ if it is tangent to $\partial M$. The set $\Omega_E^0$ contains the glancing covectors and, under \eqref{eqn:PassumeV}, has measure zero in $p^{-1}(E)\cap T^*_{\partial ^oM}M.$

We define the broken Hamiltonian flow as follows. (see, for example, \cite[Appendix A]{DyZw}). We denote this flow by $\varphi_t:=\exp(H_pt).$ Assume without loss of generality that $t>0$. We consider $\exp(tH_p)(x,\xi)$ defined on $T^*\widetilde{M}$, and let $t_0$ be the first nonnegative time when $\exp(tH_p)(x,\xi)$ hits the boundary. If this happens at a non-smooth point of the boundary, or if $\exp(t_0H_p)(x,\xi)\in \Omega_E^0$, then the flow cannot be extended past $t=t_0$. Otherwise, $\exp(t_0H_p)(x,\xi)\in \Omega_E^+$ and there exists unique $(x_0,\xi_0)\in \Omega_E^-$ such that the natural projections of $\exp(t_0H_p)(x,\xi)$ and $(x_0,\xi_0)$ onto $T^*\partial M$ are the same. We then define $\varphi_t$ inductively, by putting $\varphi_t(x,\xi)=\exp(tH_p)(x,\xi)$ for $0<t<t_0$ and $\varphi_t(x,\xi)=\varphi_{t-t_0}(x_0,\xi_0)$ for $t>t_0.$
Define for $T>0$, the set 
$$\mcal{B}_T\subset T^*M\cap p^{-1}([a,b])$$
to be the closed set of all $(x,\xi)$ such that $\varphi_t(x,\xi)$ intersects a glancing point for some $t\in[-T,T].$ Then, as shown in \cite[Lemma 1]{ZeZw}, $\mcal{B}_T$ has measure zero in $p^{-1}(E)$.

In many cases, quasimodes with specified concentration properties are easier to construct than corresponding eigenfunctions. One example of this is Hassell's use of concentrating quasimodes in \cite{Hass} to show non quantum unique ergodicity in some domains with ergodic billiard flow (for another, see Section \ref{sec:example}). Other connections between quasimodes and quantum ergodicity have been explored in \cite{AnanRivi}, \cite{MarklofLeak}, and \cite{zeldNote}. The utility of quasimodes motivates us to generalize existing quantum ergodicity results to that setting ( Section \ref{sec:QEquasimodes}). Although this is a natural generalization and not unexpected, it does not seem to be available in the literature.  

We make the following definitions (with the obvious analog in the homogeneous setting under the rescaling $h=\lambda^{-1}$).
\begin{defin}
\label{def:posDens}
A positive density set of quasimodes for $P$ on $[a,b]$, $0\leq a<b$ is a collection $\{(\psi_j,E_j),j=1,...\}$ satisfying,
\numList
\item $\|\psi_j\|_{L^2}=1$
\item $\|(P-E_j)\psi_j\|_{L^2}=o(h^{2d+1})$
\item $|\la \psi_j,\psi_k\ra| = o(h^{2d})\quad j\neq k$
\item $|\{E_j\in [a,b]\}|\geq ch^{-d},\, c>0.$
\enumList
\end{defin}
\begin{defin}
\label{def:complete}
Let $E_j,j=1,...$ be all the eigenvalues of $P$ with multiplicity. A complete set of quasimodes for $P$ is a collection $\{\psi_j, j=1,...\}$ satisfying 

$$\|\psi_j\|_{L^2}=1,\quad\|(P-E_j)\psi_j\|_{L^2}=o(h^{2d+1}),\quad|\la\psi_j,\psi_k\ra|=o(h^{2d}), \text{ } j\neq k.$$

\end{defin}

\vbox{\noindent{\bf Remarks: }
\begin{enumerate}
\item Notice that the set of all eigenfunctions is a complete set of quasimodes in the sense of Definition \ref{def:complete}, in particular, a positive density set of quasimodes for $P$ on $[a,b]$. Hence, our results apply to eigenfunctions.
\item The discrepancies for these quasimodes are much smaller than the mean spectral density given by the Weyl law and are necessary so that Hilbert-Schimidt norms can be accurately expressed in terms of quasimodes.
\item Note that, although it is often easy to construct positive density sets of quasimodes with specified concentration properties, it is difficult to construct complete sets of quasimodes with the same. However, the proofs of our results are simplified by using complete sets of quasimodes. To avoid these difficulties, we formulate our results in terms of positive density sets of quasimodes. and prove in Lemma \ref{lem:completeQuas} that it is enough to consider complete sets of quasimodes.
\end{enumerate}
}

\noindent Under assumptions \eqref{eqn:assumeE}, \eqref{eqn:PassumeV}, and 
\begin{equation}
\label{eqn:flowAssume} \varphi_t \text{ is ergodic on } p^{-1}(E),\quad E\in[a,b]
\end{equation}
we have the following analog of quantum ergodicity for a positive density set of quasimodes for $P$. For a semiclassical pseudodiffential operator of order $m$, we write $B\in \Psi_h^m$, and denote its semiclassical symbol, $\sigma_h(B)$ (see, for example, \cite[Section 14.2]{EZB} for definitions).

\begin{theorem}
\label{thm:QEquasi}
Suppose that $(M,g)$ is a compact manifold with a piecewise smooth boundary and \eqref{eqn:flowAssume} holds. Let $\{(\psi_j,E_j),j=1,...\}$ be a positive density set of quasimodes of the Dirichlet realization of $P$ on $[a,b]$ for $0<a<b$. Then there exist a family of subsets 
$$\Lambda(h)\subset \{\psi_j|a\leq E_j\leq b\}$$
of full density such that if $\{\psi_k\}\subset \Lambda(h)$ is a sequence with unique defect measure $\mu$ then there is an $E\in [a,b]$ such that $\mu\equiv (\mu_E(p^{-1}(E)))^{-1} \mu_E.$
\end{theorem}

\noindent Here a defect measure for a sequence $\psi_j$ is the weak limit of the semiclassical measures $\mu_j$ associated to $\psi_j$  (See, for example, \cite[Chapter 5]{EZB} for details.)

\vbox{\noindent{\bf Remarks: }
\begin{enumerate}
\item An equivalent way of formulating this theorem is that for all $\psi_j\in \Lambda(h)$ and $B\in \Psi_h^0(M^o)$ with symbol, $\sigma_h(B)$, compactly supported away from $\partial M$, 
$$\left|\la B\psi_j,\psi_j\ra - \fint_{p^{-1}(E_j)}\sigma_h(B)d\mu_{E_j}\right|\to 0.$$
\item Note that Theorem \ref{thm:QEquasi} applies to positive density subset of eigenfunctions in energy bands $E_j\in [a,b]$, while that in \cite{Helff} is stronger and applies to such subsets of energy shells $E_j\in [a-Ch,a+Ch].$ However, \cite{Helff} does not apply to manifolds with boundary and such an energy shell quantum ergodicity theorem is not known in that case. It would be interesting to obtain such a result.
\end{enumerate}
}

In section \ref{sec:QEsubsets}, we relax the dynamical assumption \eqref{eqn:flowAssume} and instead make the mixed dynamical hypothesese. 
\begin{equation}
\begin{array}{cccc}
\label{eqn:AssumeSubset}
\exists\text{ } U\subset p^{-1}([a,b]),&\mu_E(U\cap p^{-1}(E))>0,&
\mu_E(\partial U\setminus U\cap p^{-1}(E))=0,&E\in [a,b]
\end{array}
\end{equation}
where $\mu_E$ is the Liouville measure on $p^{-1}(E)$ and 
\begin{equation}
\label{eqn:assumeFlowSubset}
\begin{array}{cc}
\varphi_t \text{ is ergodic on } U\cap p^{-1}(E), E\in[a,b]
&
\varphi_{t}(U)=U.\\
\end{array}
\end{equation} 

\noindent{\bf Remark:} Hypotheses \eqref{eqn:AssumeSubset} include the case where $U$ is closed with a boundary of positive measure.

In this case, we have the following result:
\begin{theorem}
\label{thm:QESubset}
Suppose that $(M,g)$ is a compact manifold with a piecewise smooth boundary, and that \eqref{eqn:assumeE}, \eqref{eqn:PassumeV}, \eqref{eqn:AssumeSubset}, and \eqref{eqn:assumeFlowSubset} hold. Let $\{(\psi_j,E_j), j=1,...\}$ be a positive density set of quasimodes of the Dirichlet realization of $P$ on $[a,b]$, $0<a<b$.
Then there exist a family of subsets 
$$\Lambda(h)\subset \{\psi_j|a\leq E_j\leq b\}$$
of full density such that if $\{\psi_k\}\subset \Lambda(h)$ is a further subsequence such that $\{\psi_k\}$ has unique defect measure, $\mu$, then there exist $E\in [a,b]$ and $c$ such that supp $\mu\subset p^{-1}(E)$ and $\mu|_{U\cap p^{-1}(E)}\equiv c\mu_E|_{U\cap p^{-1}(E)}$.
\end{theorem}

\noindent{\bf Remark:} Although Theorem \ref{thm:QEquasi} can be seen as a corollary of Theorem \ref{thm:QESubset}, the proof of Theorem \ref{thm:QEquasi} is essential to that of Theorem \ref{thm:QESubset} and, in addition, presenting the direct proof is useful to demonstrate the new ideas involved in obtaining Theorem \ref{thm:QESubset}.

In section \ref{sec:semiclassToClass}, we specialize to the case $P=-h^2\Delta$ and pass from the semiclassical quantization to the Kohn-Nirenberg calculus. For a homogeneous pseudodifferential operator of order $m$, we write $B\in \Psi^m$ and denote its symbol, $\sigma(B)$ (see, for example, \cite[Chapter 18]{Ho3} for definitions).
\begin{theorem}
\label{thm:QEquasiClassical}
Suppose that $(M,g)$ is a compact manifold with a piecewise smooth boundary and \eqref{eqn:flowAssume} holds. Let $\{(\psi_j,E_j),j=1,...\}$ be a positive density set of quasimodes of the Dirichlet realization of $-\Delta$ on $[0,1]$. Then, there is a full density subsequence $\psi_{j_k}$ such that for all pseudodifferential operators $B\in \Psi^0(M^o)$ with symbols compactly supported away from $\partial M$, 
$$\la B\psi_{j_k},\psi_{j_k}\ra \to \fint_{S^*M}\sigma(B)d\mu_{L}.$$
\end{theorem}
\noindent
{\bf Remark:} Although the passage from Theorem \ref{thm:QEquasi} to \ref{thm:QEquasiClassical} is known in the case of eigenfunctions and is not unexpected for quasimodes, we include it here becuase neither the result for quasimodes nor the passage from the semiclassical to the homogeneous case seem to appear in the literature.

We also specialize Theorem \ref{thm:QESubset} to the homogeneous setting
\begin{theorem} 
\label{thm:QESubsetClass}
Suppose that $(M,g)$ is a compact manifold with a piecewise smooth boundary and that there exists 
$$\begin{array}{ccccc}
U\subset S^*M,&\mu_L(U)>0, &\mu_L(\partial U\setminus U)=0, &\varphi_t\text{ is ergodic on }U,&\varphi_t(U)=U\end{array}.$$
Let $\{(\psi_j,E_j),j=1,...\}$ be a positive density set of quasimodes of the Dirichlet realization of $-\Delta$ on $[0,1]$. Then there is a full density subsequece of $\psi_j$ such that if any further subsequence $\psi_{j_k}$ has unique defect measure $\mu$, then $\mu|_U=c\mu_L|_U$ for some $c$. 
\end{theorem}

\noindent{\bf Remark:} In the homogeneous and boundaryless case, Rivi\'{e}re \cite{Riviere} recently proved a theorem on accumulation points of semiclassical measures that gives similar results for separated phase spaces. He also provides examples of separated phase spaces for the geodesic flow on boundaryless manifolds. 

Percival \cite{perc} made conjectures on eigenfunctions that were numerically verified by Barnett and Betcke in \cite{barn} in the case of mushroom billiards. In particular, Percival used the quantum-classical correspondence principle to state that most eigenfunctions concentrate either entirely in the ergodic portion of phase space or in the integrable portion and, moreover, that the eigenfunctions concentrating in the ergodic portion of phase space spread uniformly over that piece of phase space. Although we are not able to show the dichotomy between integrable and chaotic eigenfunctions, Theorems \ref{thm:QESubsetClass} and \ref{thm:QESubset} provide rigorous proofs of the fact that most eigenfunctions spread uniformly in the ergodic portions of phase space.

Finally, in Section \ref{sec:example}, we apply our results to the Dirichlet Laplacian for mushroom billiards. Figure \ref{f:1} shows two billiards trajectories, one in the ergodic region and one in the integrable region.  (For a complete description of the billiard map on mushrooms see \cite{Bunim}.)  \\

\noindent
{\sc Acknowledgemnts.} The author would like to thank Maciej Zworski for valuable discussion, St\'{e}phane Nonnenmacher and Ze\'{e}v Rudnick for pointing him toward mushroom billiards and \cite{Rudnick}, Semyon Dyatlov for advice on boundary value problem quantum ergodicity, Alex Barnett for allowing him to use the images in Figure \ref{f:2} and referring him to \cite{perc}, and the anonymous reviewer for many helpful comments. The author is grateful to the Erwin Schr\"{o}dinger Institue for support during the Summer School on Quantum Chaos and to the National Science Foundation for support under the National Science Foundation Graduate Research Fellowship Grant No. DGE 1106400, and for support during the Erwin Schr\"{o}dinger Institute's Summer School on Quantum Chaos.

\section{Quantum Ergodicity for Quasimodes}
\label{sec:QEquasimodes}
In this section, we demonstrate how to adapt quantum ergodicity results for eigenfunctions to positive density sets of quasimodes. 

We will need the following lemmas relating various spectral quantities for operators to their corresponding expressions involving complete sets of quasimodes. 
\begin{lemma}
\label{lem:repEigenFunc}
Let $\{u_k, k=1,...\}$ be the eigenfunctions of $P$ corresponding to $E_k$ and let $\{\psi_j, j=1,...\}$ be a complete set of quasimodes for $P$. Then,
$$u_j=\sum_{|E_j-E_k|<Ch^{d+1}}c_{jk}\psi_k+o(h^{d}).$$
\end{lemma}
\begin{proof}
First, observe that, near an energy level $E$, there exists $(c_1(h),c_2(h))$ with $|c_1(h)-c_2(h)|\geq ch^{d+1}$ such that  $((E-c_2(h),E-c_1(h))\cup(E+c_1(h),E+c_2(h)))\cap\{E_j,j=1,...\}=\emptyset$ -- if this were false, then the spectrum of $P$ would violate the Weyl law.

Now, let $\Lambda=\{\psi_j|E_j\in (E-c_1(h),E+c_1(h))\}$ and $\Pi$ be the spectral projection onto $[E-c_1(h),E+c_1(h)]$. By the Weyl law, $|\Lambda|\leq Ch^{-d}$. Therefore, by \cite[Proposition 32.4]{Lazutkin}, $$\dim\text{ }(\Pi \text{ span }(\Lambda))=|\Lambda|,$$ since $$o(h^{2d})+o(h^{2d+1})O(h^{-d-1})=o(h^d).$$ But, rank $\Pi=|\Lambda|$. Hence, span $ \Lambda$ = range $\Pi$ and the result follows from the almost orthogonality of $\psi_j$.
\end{proof}

\begin{lemma}
Let $A$ be a Hilbert-Schmidt operator and $\{\psi_j,j=1,...\}$ be a complete set of quasimodes for $P$ and $a<b$. Then,
\begin{equation}
\label{eqn:HS}
h^d\sum_{a\leq E_j\leq b}\|A\psi_j\|^2=h^d\|\Pi_{[a,b]}A\|^2_{HS}+o(1)\end{equation}
and if $A$ is of trace class, 
\begin{equation}\label{eqn:Trace}h^d\sum_{a\leq E_j\leq b}\la A\psi_j,\psi_j\ra=h^d\text{Tr}\Pi_{[a,b]}A+o(1).\end{equation}
\label{lem:HS}
\end{lemma}
\begin{proof}
First, let $(c_{jk})=\la u_k,\psi_j\ra$. Then, by Lemma \ref{lem:repEigenFunc},
$$u_j=\sum_kc_{jk}\psi_k+o(h^{d})\quad \psi_j=\sum_k\overline{c_{jk}}u_k,$$
where all sums are taken over $k$ such that $|E_k-E_j|<Ch^{d+1}$.
Observe that 
$$u_j=\sum_kc_{jk}\psi_k+o(h^d)=\sum_{k,k'}c_{jk}\overline{c_{kk'}}u_{k'}+o(h^d).$$
Hence, by the orthogonality of $u_j$,
$$\sum_{k}c_{jk}\overline{c_{kk'}}=\delta_{kj}+o(h^d).$$

Now, 
\begin{eqnarray}
h^d\sum_{a\leq E_j\leq b}\la A\psi_j,A\psi_j\ra&=&h^d\sum_{j,k,k'}c_{jk}\overline{c_{jk'}}\la Au_k,Au_k'\ra =h^d\sum_{k,k'}(\delta_{kk'}+o(h^{d}))\la Au_k,Au_{k'}\ra \nn 
&=&h^d\|\Pi_{[a,b]}A\|^2_{HS}+h^d\sum_{k,k'}o(h^{d})\la Au_k,Au_{k'}\ra=h^d\|\Pi_{[a,b]}A\|^2_{HS}+o(1)\nonumber
\end{eqnarray}
where the last equality follows from the fact that there are at most $O(h^{-d})$ terms in each sum.

\noindent An analgous argument shows that \eqref{eqn:Trace} holds. 
\end{proof}

To prove Theorem \ref{thm:QEquasi}, we need the following lemma similar to \cite[Lemma A.2]{DyZw}
\begin{lemma}
\label{lem:CutoutBdry}
Let $\chi\in C_c^\infty(M^o)$ with $0\leq \chi \leq 1.$ Then for $a'<a<b<b'$,
$$(2\pi h)^d\sum_{E_j\in[a,b]}\int_M(1-\chi)|\psi_j|^2d\textup{Vol}\leq \int_{T^*M\cap p^{-1}([a',b'])}(1-\chi )d\mu_\sigma+o(1)\text{ as }h\to 0.$$
\end{lemma}
\begin{proof}
This follows directly from \cite[Lemma A.2]{DyZw} and Lemma \ref{lem:HS} once we observe that 
$$(2\pi h)^d\sum_{E_j\in[a,b]}\int_M(1-\chi)|\psi_j|^2d\textup{Vol}=(2\pi h)^d\sum_{E_j\in [a,b]}\la (1-\chi)\psi_j,\psi_j\ra=(2\pi h)^d\text{Tr }\Pi_{[a,b]} (1-\chi) +o(1).$$
\end{proof}

The following lemma allows us to take arbitrary collections of orthogonal quasimodes and complete them. Hence, we only have to prove results for complete sets of quasimodes and they follow for positive density sets.
\begin{lemma}
\label{lem:completeQuas}
Let $\{(\psi_j,E_j), j=1,...,J\}$ be a set of almost orthogonal quasimodes with 
$$\|\psi_j\|_{L^2}=1,\quad \|(P-E_j)\psi_j\|_{L^2}=o(h^{2d+1}),\quad |\la \psi_j,\psi_k\ra|=o(h^{2d}),\text{ }j\neq k.$$
 Then, there exists $\{\varphi_k$, $k=1,...K\}$ such that the set $\{\varphi_k, k=1,...,K\}\cup\{\psi_j, j=1,...,J\}$ is a complete set of quasimodes for $P$. 
\end{lemma}
\noindent {\bf Remark:} Note that this lemma applies to sets of quasimodes that do not have positive density. 
\begin{proof}
Let $u_j$ be eigenfunctions of $P$ corresponding to $E_j$. Select an energy $E$. Then, by the proof of Lemma \ref{lem:repEigenFunc}, there are gaps at distance $c_1(h)$ of size $h^{d+1}$ in the spectrum of $P$ near $E$. Let $\Lambda=\{u_k:|E_j-E|\leq c_1(h)\}$.  Let $\Lambda ':=\{\psi_j:|E_j-E|= o(h^{d+1})\}$. Then, for $\psi_j\in \Lambda '$, we have that 
$$\psi_j=\sum_{u_k\in \Lambda}c_{jk}u_k.$$
Now, define $N:=|\Lambda|$ and $M=|\Lambda '|$ letting $b_j=(c_{j1},c_{j2},...,c_{jN})$, $j=1,..M$. Let $\{e_1,...,e_m\}$ be an orthonormal basis for span $(\{b_j:j=1,...,M\})$.  Apply the Gram-Schmidt process to obtain an orthonormal basis $\{e_1,...,e_M,v_{M+1},...,v_N\}$ of $\re^N$ where 
$$v_k=(v_{k1},...,v_{kN}).$$
Then, letting 
$$\varphi_k=\sum_{j}v_{kj}u_j,\quad M+1\leq k\leq N,$$
 $\{\psi_1,...,\psi_M,\varphi_{M+1},...,\varphi_N\}$ is an almost orthonormal basis for span $\Lambda$. Repeating this process for each cluster, we obtain a complete set of quasimodes. 
\end{proof}

We also need the following restatement of results in \cite[Section 4.3]{christianson} that is found in \cite[Lemma A1]{DyZw}. 
\begin{lemma}
\label{lem:FIO}
Fix $T>0$. Assume that $A\in \Psi^{comp}(M^o)$ is supported away from the boundary of $M$ and $WF_h(A)\subset p^{-1}([a,b])\setminus \mcal{B}_T.$ Then, for each $\chi\in C_c^\infty(M^o)$ and for each $t\in[-T,T]$, the operator $\chi e^{-itP/h}A$ is a Fourier integral operator supported away from $\partial M$ and associated to the restriction of $\varphi_t$ to a neighborhood of $WF_h(A)\cap\varphi_t^{-1}(\text{supp }\chi)$, plus an $O_{L^2\to L^2}(h^\infty)$ remainder. The following version of Egorov's Theorem holds:
$$\chi e^{itP/h}Ae^{-itP/h}=A_{t,\chi}+O(h^\infty)_{L^2(M)\to L^2(M)},$$
where $A_{T,\chi}\in \Psi^{comp}(M^o)$ is supported away from $\partial M$ and $\sigma_h(A_{T,\chi})=\chi(a\composed \varphi_t).$
\end{lemma}

Now, we prove Theorem \ref{thm:QEquasi}, following \cite[Appendix A]{DyZw}
\begin{proof}
Let $(\psi_j,E_j)$ be a positive density set of quasimodes. By Lemma \ref{lem:completeQuas}, and the fact that our original quasimodes have positive density, we may assume without loss that $(\psi_j,E_j)$ is a complete set of quasimodes (see Definition \ref{def:complete}). 

We first show that 
$$\limsup_{h\to 0}h^d\sum_{E_j\in [a,b]}\left|\la B\psi_j,\psi_j\ra -\fint_{p^{-1}(E_j)}\sigma_h(B)d\mu_{E_j}\right|=0.$$ Then, by a standard diagonal argument that can be found, for example, in \cite[Theorem 15.5]{EZB}, we extract the set $\Lambda(h)$.

Take $a',b'$ such that $a'<a<b<b'$ and the assumptions \eqref{eqn:PassumeV} and \eqref{eqn:assumeE} hold for $E\in[a',b']$. (If they do not hold when $E\notin[a,b]$, then we need to take $a',b'$ getting close to $a$ and $b$. e.g. Let $a'=a-1/T$ and $b'=b+1/T.$ Then estimate the extra contributions by the Weyl law.) Now, fix a $T>0$ and choose $\chi_T\in C_c^\infty(M^o)$ with $0\leq \chi_T\leq 1 $ and 
$$\int_{T^*M\cap p^{-1}([a'-1,b'-1])}(1-\chi_T)^2d\mu_\sigma\leq T^{-1}.$$
Let $\psi\in C_c^\infty(a'-1,b'+1)$ have 
$$\psi(E)\int_{p^{-1}(E)}\chi_Td\mu_E=\int_{p^{-1}(E)}\sigma_h(B)d\mu_E,\text{  } E\in[a',b'].$$
Then, by Lemma \ref{lem:CutoutBdry}, it is enough to show that the conclusion holds for $B-\psi(P(h))\chi_T$, whose symbol integrates to 0 on $p^{-1}(E).$ Therefore, we assume, without loss, that 
$$\int_{p^{-1}(E)}\sigma_h(B)d\mu_E=0,\text{  } E\in[a',b'].$$

By elliptic estimates, we may assume that $WF_h(B)\subset p^{-1}((a',b'))$. More specifically, $B\in \Psi^{comp}.$ Thus, we can write $B=B_T'+B_T''$, where $WF_h(B_T')\cap \mcal{B}_T=\emptyset$ and $\|\sigma_h(B_T'')\|_{L^2(p^{-1}[a',b'])}\leq T^{-1}.$ 

Then, by H\"{o}lder's inequality, Lemma \ref{lem:HS} and \cite[Lemma 2.2]{DyZw}, we have that 
\begin{eqnarray}
h^d\sum_{E_j\in[a,b]}|\la B_T''\psi_j,\psi_j\ra|&\leq& C\left(h^d\sum_{E_j\in[a,b]}\| B_T''\psi_j\|^2\right)^{1/2}\nn
&=&C\left(h^d\sum_{E_j\in[a,b]}\| B_T''u_j\|^2\right)^{1/2}+o(1)\nn 
&\leq& C\|\sigma_h(B_T'')\|_{L^2(p^{-1}[a',b'])}+o(1).\nonumber 
\end{eqnarray}
Hence, the contribution of $B_T''$ goes to 0 in the limit $\lim_{T\to\infty}\limsup_{h\to 0}$ and we may replace $B$ by $B_T'$. 

Now, by Duhamel's formula and the unitarity of $e^{itP/h}$,
$$e^{itP/h}\psi_j=e^{itE_j/h}\psi_j+\frac{i}{h}\int_0^te^{i(t-s)P/h}o(h^{2d+1})=e^{itE_j/h}\psi_j+o_T(h^{2d}).$$
So, defining
$$\la A\ra_T:=\recip{T}\int_0^Te^{itP/h}Ae^{-itP/h}dt,$$
and using Lemma \ref{lem:HS} and H\"{o}lder's inequality, we have 
\begin{eqnarray}
h^d\sum_{E_j\in[a,b]}|\la B_T'\psi_j,\psi_j\ra|&=&h^d\sum_{E_j\in[a,b]}|\la B_T'e^{-itE_j/h}\psi_j,e^{-itE_j/h}\psi_j\ra|\nn 
&=&h^d\sum_{E_j\in[a,b]}|\la \la B_T'\ra_T\psi_j,\psi_j\ra|+o_T(h^{2d})\nn 
&\leq &\left(h^d\sum_{E_j\in[a,b]}(\| \la B_T'\ra_T\psi_j\|+o_T(h^{2d}))^2\right)^{1/2}
\nn
&=&\left(h^d\sum_{E_j\in [a,b]}\| \la B_T'\ra_Tu_j\|^2\right)^{1/2}+o_T(1)\nonumber
\end{eqnarray}

From this point forward, the proof is identical to that in \cite[Theorem A.2]{DyZw}.
By Lemma \ref{lem:FIO}, $\la B_T'\ra_T\chi_T$ is, up to an $O(h^\infty)_{L^2\to L^2}$ remainder, a pseudodifferential operator in $\Psi^{comp}$ compactly supported inside $M^o$ and with principal symbol
$$\sigma_h(\la B_T'\ra_T\chi_T)=\frac{\chi_T}{T}\int_0^T\sigma_h(B_T')\composed \varphi_tdt.$$

Now, all that remains to show is 
$$\lim_{T\to\infty}\limsup_{h\to 0}h^d\sum_{E_j\in [a,b]}\| \la B_T'\ra_Tu_j\|^2=0.$$
Since, by Lemma \ref{lem:CutoutBdry}, 
$$\limsup_{h\to 0}(2\pi h)^d\sum_{E_j\in[a,b]}\|(1-\chi_T)u_j\|_{L^2}^2\leq T^{-1},$$
we can replace $\la B_T'\ra_T$ by $\la B_T'\ra \chi_T.$ Thus, by \cite[Lemma 2.2]{DyZw}, it remains to show that 
\begin{equation}
\label{eqn:ergodicApp}\lim_{T\to\infty}\|\sigma_h(\la B_T'\ra_T\chi_T)\|_{L^2(p^{-1}([a',b']))}\leq \lim_{T\to \infty} \|\la \sigma_h(B_T')\ra_T\|_{L^2(p^{-1}([a',b']))}=0.\end{equation}
To do this, write 
$$\|\la \sigma_h(B_T')\ra_T\|_{L^2(p^{-1}(E))}\leq \|\la \sigma_h(B)\ra_T\|_{L^2(p^{-1}(E))}+\|\la \sigma_h(B''_T)\ra_T\|_{L^2(p^{-1}(E))}.$$
The first term on the right goes to 0 when $T\to \infty$ by the von Neumann ergodic theorem and the second term is bounded by $\|\sigma_h(B_T'')\|_{L^2(p^{-1}(E))}$ and hence also goes to $0$. 
\end{proof}

\noindent{\bf Remark:} Notice that \eqref{eqn:ergodicApp} is the only step in which the ergodicity of the flow is used. This will be important when we adapt the result to ergodic invariant subsets of phase space.

\section{Quantum Ergodicity for Subsets}
\label{sec:QEsubsets}
Now, we prove Theorem \ref{thm:QESubset} using the fact that only step \eqref{eqn:ergodicApp} uses the ergodicity of the flow. For completeness, we include the following elementary lemma. 

\begin{lemma}
\label{lem:measures}
Let $U$ satisfy \eqref{eqn:AssumeSubset}, let $\mu_2$ be a finite measure on $p^{-1}(E)$ and suppose that for all $a\in C^\infty(p^{-1}(E))$ compactly supported away from $\partial M$ with $\int_{U\cap p^{-1}(E)} ad\mu_E=0$, we have $\int_{U\cap p^{-1}(E)} ad\mu_2=0$. Then, $\mu_2|_{U\cap p^{-1}(E)}\equiv c\mu_E|_{U\cap p^{-1}(E)}$ for some $c\geq 0$.
\end{lemma}
\begin{proof}
Let $V=U\cap p^{-1}(E).$ Let $\mu_1=(\mu_E(V))^{-1}\mu_E$. Then, $\mu_1(V)=1$. Define $\chi\in C^\infty(p^{-1}(E))$ compactly supported away from $\partial M$ with $0\leq \chi\leq 2$ and $\int_V \chi d\mu_1 =1$ and $\int_V \chi d\mu_2=c_2>0$. (To see that such functions exist simply take non-negative approximations to $1_{V}$ with support inside $M^o$.)

%We first show $\mu_2<<\mu_1$. Suppose there is a set $A\subset U$ such that $\mu_2(A)>0$ and $\mu_1(A)=0$. Then, let $a\in C^\infty(S^*M)$ with $\int_U ad\mu_1=0$. Let $O_\e\supset A$ have $\mu_1(O_\e)<\e$. Then, let $b_\e\in C^\infty (S^*M)$ have $b_\e\geq 0$, $b_\e\equiv 1 $ on $O_\e$, and
%$$\|b_\e-1_{O_\e}\|_{L^1(\mu_1)}<\e.$$
%Then, let $\chi_\e:=\|b_\e\|_{L^1(\mu_1)}\chi$. Then $0\leq \chi_\e\leq 2\e \chi$, and we have that $\int (b_\e-\chi_\e)d\mu_1=0$ and hence 
%$$0=\int_U(b_\e-\chi_\e)d\mu_2\geq \mu_2(A)-\|\chi_\e\|_\infty \mu_2(U)>\mu_2(A)-2\e c_2>0$$
%for $\e$ small enough. Hence $\mu_2<<\mu_1$. 

Now, let $a\in C^\infty(p^{-1}(E))$ compactly supported away from $\partial M$ with $\int_V ad\mu_1=\bar{a}$. Then, $\int_V a-\bar{a}\chi d\mu_1=0$ and hence $\int_V a-\bar{a}\chi d\mu_2=0$. Therefore, for $a\in C^\infty(p^{-1}(E))$ compactly supported away from $\partial M$, 
$$\int_V ad\mu_2=c_2\int_V ad\mu_1.$$ 
But, $\mu_1|_V$ and $\mu_2|_V$ are positive distributions of order 0 since V satisfies \eqref{eqn:AssumeSubset}. Hence, since
$$\int_V ad\mu_2=\int_Vad(c_2\mu_1)=0$$
for all $a\in C^\infty(p^{-1}(E))$ compactly supported away from $\partial M$, and $M$ can be exhausted by compact subsets, $\mu_2|_V\equiv c_2\mu_1|_V.$
\end{proof}

We now prove Theorem \ref{thm:QESubset}.
\begin{proof}
Lemma \ref{lem:completeQuas} shows that without loss of generality we may work with $\psi_j$ forming a complete set of quasimodes.
Fix $A\in \Psi_h^0(M)$ with symbol $a=\sigma_h(A)$ compactly supported away from $\partial M$ satisfying
\begin{equation}\label{eqn:avgzero}\int_{U\cap p^{-1}(E)}ad\mu_E=0,\quad E\in [a,b].\end{equation} 
Fix $\e>0$. Then, let $U_\e$ be open with $\overline{U}\subset U_\e$ and $\mu_E((U_\e\setminus U)\cap p^{-1}(E))<\e$ for $E\in [a,b]$ (note that we use \eqref{eqn:AssumeSubset} here).
Let $\chi_\e\in C_c^\infty(U_\e)$ compactly supported away from $\partial M$ with $\chi_\e|_{\bar{U}}\equiv 1$, $0\leq \chi_\e\leq 1$. Then, let $a_\e=\chi_\e a$. 
Based on the arguments used to prove Theorems \ref{thm:QEquasi}, we have that 
\begin{equation}\label{eqn:subsetSum}\limsup_{h\to 0} h^{d}\sum_{E_j\in [a,b]}\left|\la A_\e\psi_j,\psi_j\ra\right|\leq C\|\la a_\e\ra_T\|_{L^2(p^{-1}[a,b])},\end{equation}
 where $A_\e=a_\e(x,hD)$,
$$ \la a_\e\ra_T:=\recip{T}\int_{0}^Ta_\e\composed \varphi_t(x,\xi)dt.$$

We have that $\varphi_t$ is ergodic on $U$ and $U$ is invariant under $\varphi_t$.
Hence, by the von Neumann ergodic theorem $\la a_\e 1_U\ra_T\to-\!\!\!\!\!\int_Ua_\e d\mu_E= 0$ in $L^2(p^{-1}(E))$. Again, by the ergodic theorem, $\la a_\e 1_{U^c}\ra_T\to Pa_\e$ in $L^2(p^{-1}(E))$ with $\|Pa_\e\|_{L^2}\leq C
\|a\|_{L^\infty}\e.$
Hence,
\begin{equation}\label{eqn:subsetRHS}\lim_{T\to \infty}\|\la a_\e\ra_T\|_{L^2(p^{-1}[a,b])}\leq C\|a\|_{L^\infty}\e.\end{equation}

We now show that there exists a full density subset 
$$\Lambda(h)\subset \{(\psi_j,E_j)|a\leq E_j\leq b\}$$
such that for all $\psi_j\in \Lambda(h)$ and all  $A\in \Psi_h^0$ with symbol $a=\sigma_h(A)$ satisfying \eqref{eqn:avgzero}
\begin{equation}\label{eqn:epsto0}
\lim_{\e \to 0}\lim_{h\to 0}\la A_\e \psi_j,\psi_j\ra=0.\end{equation}

We first do this for one such $a\in C^\infty(T^*M)$. Fix $\e>0$ and let 
$$\Gamma(h,\e)=\{a\leq E_j\leq b:|\la A_\e \psi_j,\psi_j\ra|\geq (C\|a\|_{L^\infty}\e)^\recip{2}\}.$$
Then, by the Chebyshev inequality, \eqref{eqn:subsetSum}, and \eqref{eqn:subsetRHS},
$$h^d|\Gamma(h,\e)|\leq (C\e\|a\|_{L^\infty})^{\recip{2}}$$ 
and for $E_j\notin \Gamma(h)$
$$|\la A_\e \psi_j,\psi_j\ra|\leq (C\e\|a\|_{L^\infty})^{\recip{2}}.$$
But, by the Weyl law,
$$|\Gamma(h,\e)|/|\{a\leq E_j\leq b\}|=h^d|\Gamma(h,\e)|/\left(\text{Vol}(p^{-1}[a,b])+o(1)\right)\leq C(\|a\|_{L^\infty}\e)^{1/2}+o(1).$$
Now, let $\e_m\to 0$ and define $\Gamma(h)=\cap_m\Gamma(h,\e_m)$. Then, $\Gamma(h)$ has $0$ density as $h\to 0$. Hence, $\Lambda(h)=\{a\leq E_j\leq b\}\setminus \Gamma(h)$ has full density as $h\to 0$ and has the desired property for the operator $a(x,hD)$. 

To complete the construction of $\Lambda(h)$, first observe that it is enough to consider $A\in \Psi_h^{-\infty}$ since $\psi_j$ are microlocalized in $p^{-1}[a,b]$. Then, take a countable dense set, $\{A_k\}$ of $\Psi_h^{-\infty}\cap\{A\in\Psi_h^0:\|\sigma_h(A)\|_{L^\infty}\leq 1,\text{ }\sigma_h(A)\text{ satisfies \eqref{eqn:avgzero}}\}$ and apply a variant of the standard diagonal argument (contained for example, in \cite[Theorem 15.5]{EZB}). 

Now, suppose that $\{\psi_j\}\subset \Lambda(h)$ is a further subset with defect measure $\mu$. Then, for $A\in \Psi_h^0$ with symbol satisfying \eqref{eqn:avgzero}, we have that 
$$\lim_{h\to 0}\la A_\e\psi_{j},\psi_{j}\ra= \int \sigma_h(A)_\e d\mu=o(1).$$ 
Hence, by the dominated convergence theorem,
$$\int_U\sigma_h(A)d\mu=0.$$
But, since $\{\psi_j\}$ is a subset of a positive density set of quasimodes on $[a,b]$, there exists $E\in [a,b]$ such that supp $\mu\subset p^{-1}(E).$
 Thus, we may apply Lemma \ref{lem:measures} to obtain the result.
\end{proof}

\section{From Semiclassical to Standard Quantum Ergodicity}
\label{sec:semiclassToClass}

For completeness and to present the proof in the quasimode case, we will now pass from Theorem \ref{thm:QEquasi} with $P=-h^2\Delta$ to Theorem \ref{thm:QEquasiClassical}.
\begin{proof}
Let $h^2\lambda_j^2=E_j$ and $\lambda=h^{-1}$, $\chi\in C^\infty$ $\chi(\xi)\equiv 0$ for $|\xi|\leq 1$ and $\chi\equiv 1$ for $|\xi|\geq 2$, and $\chi_\e=\chi(\xi/\e).$ 

Let $\hat{A}$ be a homogeneous pseudodifferential operator of order $0$ on $M$ with symbol compactly supported away from $\partial M$. Define $a_0:=\sigma(\hat{A})$. Then, $a_0$ is homogeneous degree $0$ on $T^*M\setminus\{0\}$. Hence $a_0(x,D)\chi(D)=a_0(x,hD)\chi(hD/h)$. Now, define $A_\e\in \Psi^0(M^o)$ by 
$$A_\e:=a(x,hD)\chi_\e(hD).$$

Theorem \ref{thm:QEquasi} gives, for $0<a<b$, 
$$h^d\sum_{h\lambda_j\in [a,b]}\left|\la A_\e \psi_j,\psi_j\ra -\fint_{p^{-1}(E_j)} \sigma_h(A_\e)d\mu_{E_{j}}\right|\to 0,\text{ }h=\lambda^{-1}\to 0.$$
But, since $\sigma_h(A_\e)=a_0(x,\xi)\chi_\e(\xi)$ and $a_0$ is homogeneous degree $0$, we have
$$\fint_{p^{-1}(E_j)}\sigma_h(A_\e)d\mu_{E_j}=\fint_{S^*M}\sigma(\hat{A})d\mu_L+O(\e).$$
Hence, we need to show,
$$\lim_{\e\to 0}\limsup_{h\to 0}h^d\sum_{h\lambda_j\in[a,b]}\la (\hat{A}-A_\e)\psi_j,\psi_j\ra =0.$$
By H\"{o}lder's inequality, Lemma \ref{lem:HS}, and \cite[Lemma 2.2]{DyZw},
\begin{eqnarray}
h^d\sum_{h\lambda_j\in[a,b]}|\la (\hat{A}-A_\e)\psi_j,\psi_j\ra|&\leq& C\left(h^d\sum_{h\lambda_j\in[a,b]}\|(\hat{A}-A_\e)\psi_j\|^2 \right)^{1/2}
\nn
 &=&C\left(h^d\sum_{h\lambda_j \in [a,b]}\|(\hat{A}-A_\e)u_j\|^2\right)^{1/2}+o(1)\nn 
&\leq& C\|\sigma_h(\hat{A}-A_\e)\|_{L^2(p^{-1}([a,b]))}+o(1).\nonumber
\end{eqnarray}
But, 
$\lim_{\e\to 0}\|\sigma_h(\hat{A}-A_\e)\|_{L^2(p^{-1}([a,b]))}=0.$

Thus, for any $\delta>0$, we have 
$$\limsup_{h\to 0}h^{d}\sum_{h\lambda_j\in[\delta,1]}\left|\la \hat{A} \psi_j,\psi_j\ra -\fint_{S^*M} \sigma(\hat{A})d\mu_{L}\right|=0.$$
All that remains is to show that 
$$\lim_{\delta\to 0}\limsup_{h\to 0}h^{d}\sum_{h\lambda_j\in[0,\delta)}\left|\la \hat{A} \psi_j,\psi_j\ra -\fint_{S^*M} \sigma(\hat{A})d\mu_{L}\right|=0.$$
But, letting $\bar{a}_0=-\!\!\!\!\!\int_{S^*M}\sigma(\hat{A})d\mu_L$, and applying the Weyl law, and H\"{o}lder's inequality,
\begin{eqnarray}
h^{d}\sum_{h\lambda_j\in[0,\delta)}|\la \hat{A} \psi_j,\psi_j\ra -\bar{a}_0|&\leq &C\left(h^{d}\sum_{h\lambda_j\in[0,\delta)}\|(\hat{A}-\bar{a}_0 )\psi_j\|^2\right)^{1/2}\nn 
&=&O(\delta^d)\|(\hat{A}-\bar{a}_0)\|_{L^2\to L^2} +O_\delta (h).\nonumber 
\end{eqnarray}

Once again, to obtain the full density subsequence, we employ a standard diagonal argument.
\end{proof}

We may also make similar arguments to those above to pass to the homogeneous setting for Theorem \ref{thm:QESubset}. In this case, we obtain Theorem \ref{thm:QESubsetClass}.

\begin{proof}
 Let 
 $$\hat{U}:=\{(x,\xi)\in T^*M:(x,\xi/|\xi|)\in U\}.$$
 Then, $\hat{U}$ satisfies \eqref{eqn:AssumeSubset}
 and \eqref{eqn:assumeFlowSubset} for any $0<a<b$. Hence, by Theorem \ref{thm:QESubset} there exists 
 $$\Lambda_\delta(h)\subset \{h\lambda_j\in[\delta, 1]\}$$
 of full density such that for all $\{\psi_j\}\subset \Lambda_\delta(h)$ with defect measure $\mu$, there is an $E\in [\delta, 1]$ such that supp $\mu\subset S_E^*M$ and $\mu|_{\hat{U}\cap S_E^*M)}\equiv c \mu_E|_{\hat{U}\cap S_E^*M}$ where 
 $$S_E^*M=\{(x,\xi)\in T^*M: |\xi|=E\}.$$
 Now, for $a\in C^\infty(T^*M\setminus\{0\})$, homogeneous of degree 0, $\e>0$ small enough and $\chi_\e$ as in the proof of Theorem \ref{thm:QEquasiClassical} 
 $$\int_{\hat{U}\cap S_E^*M}a\chi_\e d\mu_E=\int_{\hat{U}\cap S_E^*M}a d\mu_E=\int_{U}ad\mu_L$$
 and hence the homogeneous defect measure, $\mu$ of has $\mu|_U\equiv c\mu_L|_U.$
 
 But, by the Weyl Law, $h^d|\{h\lambda_j\in[0,\delta)\}|=O(\delta^d).$ Thus, letting $\Lambda(h)=\cup_{\delta>0}\Lambda_\delta(h)$, we have that $\Lambda(h)\subset \{h\lambda_j\in[0,1]\}$ has full density and for all $\{\psi_j\}$ subsequences of $\Lambda(h)$ with homogeneous defect measure $\mu$, there exists $c\geq 0$ with $\mu|_U=c\mu_L|_U.$ 
\end{proof}

\section{An example}
\label{sec:example}
We now present an example to which Theorem \ref{thm:QESubsetClass} applies. Let $\Omega$ be a symmetric mushroom billiard, as in \cite{Bunim}, composed of a hat that is a semicircle of radius 1 and a base that has width $w$ and height $h$ (figure \ref{f:1} shows such a mushroom with two billiard flows). Then $S^*\Omega$ has a subset $U$ satisfying \eqref{eqn:AssumeSubset} and \eqref{eqn:assumeFlowSubset} (see \cite{Bunim}).
Let 
$$u_{nk}=\sin (n\theta)J_n(\alpha_{n,k}r)$$
be the eigenfunctions of the Dirichlet Laplacian on the semidisk. Here $J_n$ denotes the Bessel function of the first kind of order $n$ and $\alpha_{n,k}$ the $k^{\text{th}}$ positive zero of $J_n$. 

From \cite[Appendix A]{Whisp}, we have that for $0<\gamma<\frac{2}{3}$, $z\in (0,1-n^{\gamma-\frac{2}{3}})$
$$0<J_n(nz)<2^{-n^{\gamma/3}}.$$
Also, from \cite{Olver}, we have for $0\leq k<c_2n$, that $n<\alpha_{n,k}<Cn$ for some $C>0$. Thus, $J_n(\alpha_{n,k}r)=O(2^{-n^{\gamma/3}})$ for $r<\recip{2C}$. 

Now, suppose $0<w<\recip{4C}$. Then, let $\chi\in C^\infty(\Omega)$ with $\chi\equiv 0$ in $|r|<\recip{4C}$ and $\chi\equiv 1$ in $|r|\geq \recip{2C}$. If we let $$v_{n,k}:=\chi u_{n,k},$$ for $0\leq k<c_2n$, and extend by $0$ outside of the hat of $\Omega$, then $v_{n,k}$ are quasimodes for the Dirichlet Laplacian on $\Omega$ with 
$$(-\Delta-\alpha_{n,k}^2)v_{n,k}=O(2^{-n^{\gamma/3}}).$$

Now, by the orthogonality of $u_{n,k}$, $v_{n,k}$ are orthogonal up to $O(e^{-Cn^{\gamma/3}})$. Hence, $\{v_{n,k}\}$ are a family of $\frac{c_2n(c_2n+1)}{2}$ almost orthogonal quasimodes with $O(e^{-Cn^{\gamma/3}})$ error. Thus, $$\{(v_{n,k},\alpha_{n,k}^2),0\leq k<c_2n,n=1,...\}$$ form a positive density set of quasimodes for $-\Delta$ on $[0,1]$. 

Since the $v_{n,k}$ are $O(n^{-\infty})$ quasimodes, $WF_h(v_{n,k})$ is invariant under the Hamiltonian flow (see, for example, \cite[Section 12.3]{EZB}). Therefore, since $WF_h(v_{n,k})$ does not intersect the foot of $\Omega$, the quasimodes must concentrate away from the ergodic set, $U$. Thus, for any subsequence of $\{v_{n,k}\}$ with a defect measure $\mu$, $\mu|_U\equiv 0$. Hence, Theorem \ref{thm:QESubset} applies and the constant we obtain is $0$.


\begin{thebibliography}{10}


\bibitem{AnanRivi}
N. Anantharaman and G. Rivi\'{e}re.
\newblock 
Dispersion and Controllability for the Schr\"{o}dinger EQuation and Negatively Curved Manifolds.
\newblock 
arXiv:1007.4343 .
%\bibitem{CC}
%C. Cossu and J. M. Chomaz, 
%\newblock
%Global measures of local convective instabilities.
%\newblock
%Phys. Rev. Lett. 78, 4387–4390 (1997)
\bibitem{barn}
A. H. Barnett and T. Betcke.
\newblock 
Quantum mushroom billiards.
\newblock 
Chaos 17 (2007), no. 4, 043125, 13 pp.
\bibitem{Bunim}
L. A. Bunimovich.
\newblock 
 Mushrooms and other billiards with divided phase space
\newblock
Chaos. Volume 11 number 4 (2001), 802-808. 
\bibitem{christianson}
H. Christianson. 
\newblock Quantum monodromy and non-concentration near a closed semi-hyperbolic orbit.
\newblock
Trans. Amer. Math. Soc. 363(2011), no. 7 3373-3438.
\bibitem{dever}
Y. Colin de Verdi\'{e}re.
\newblock 
 Ergodicit\'{e} et fonctions propres du laplacien.
\newblock  Commun. Math. Phys.
102, (1985) 497-502.
\bibitem{DyZw}
S. Dyatlov and M. Zworski.
\newblock
Quantum ergodicity for restrictions to hypersurfaces.
\newblock 
Nonlinearity. 26(2013), 35$ - $52. 
\bibitem{Gerr}
P. G\'{e}rard and E. Leichtnam.
\newblock
Ergodic properties of eigenfunctions for the Dirichlet problem. 
\newblock 
Duke Math. J. 71, 559 - 607 (1993)
\bibitem{Hass}
A. Hassell.
\newblock
Ergodic billiards that are not quantum unique ergodic with an appendix by the author and Luc Hillairet.
\newblock
Ann. of Math. (2) 171 (2010), no. 1, 605$ - $619. 

\bibitem{Helff}
B.Hellfer, A. Martinez, and D. Robert.
\newblock 
Ergodicit\'{e} et limite semi-classique.
\newblock 
Comm. Math. Phys. 109 (1987), no. 2, 313 - 326.

\bibitem{Ho3}
L. H\"{o}rmander,
\newblock
\emph{The Analysis of Linear Partical Differential Operators, Volume III}, 
\newblock
Springer, 1985.
\bibitem{Lazutkin}
V. F. Lazutkin.
\newblock
KAM Theory and Semiclassical Approximations to Eigenfunctions.
\newblock 
Springer-Verlag Berlin Heidelberg, USA, 1993.
\bibitem{MarklofLeak}
J. Marklof.
\newblock 
Quantum Leaks.
\newblock 
Comm. in Math. Physics. 264, Number 2 (2006), 303-316.
\bibitem{QESeparated}
J. Marklof and S. O$'$Keefe.
\newblock
Weyl$'$s law and quantum ergodicity for maps with
divided phase space.
\newblock
 Nonlinearity 18 (2005), no. 1, 277$-$304. 
\bibitem{Rudnick}
J. Marklof and Z. Rudnick, 
\newblock 
Almost all eigenfunctions of a rational polygon
are uniformly distributed.
\newblock 
J. Spectr. Theory 2 (2012), 107$-$113.
\bibitem{Whisp}
B.T. Nguyen and D. S. Grebenkov.
\newblock
Localization of Laplacian eigenfunctions in circular, spherical and elliptical domains. 
\newblock 
arXiv:1203.5022v1 22 Mar 2012.
\bibitem{Olver}
F. W. J. Olver,
\newblock
The asymptotic expansion of Bessel functions of large order.
\newblock 
Phil. Trans. of the Royal Society of London. Series A, Mathematical and Physical Sciences.
\newblock
Vol. 247, No. 930 (1954), 328-368.
\bibitem{Riviere}
G. Rivi\'{e}re.
\newblock 
Remarks on Quantum Ergodicity.
\newblock
arXiv:1209.3576.
\bibitem{perc}
I. C. Percival.
\newblock 
Regular and Irregular Spectra.
\newblock
 J. Phys. B
\newblock 
pp. L229-232 (1973).
\bibitem{schnir}
A. Shnirelmann.
\newblock Ergodic properties of eigenfuncions.
\newblock Usp. Math. Nauk. 29, (1974) 181$-$182.
\bibitem{zeldNote}
S. Zelditch.
Note on Quantum Unique Ergodicity.
\newblock
Proc. Amer. Math. Soc. 132 (2004), 1869-1872.
\bibitem{zeld}
S. Zelditch.
\newblock 
Uniform distribution of eigenfunctions on compact hyperbolic surfaces. 
\newblock Duke Math. J. 55, (1987) 919-941.
\bibitem{ZeZw}
S. Zelditch and M. Zworski.
\newblock 
 Ergodicity of eigenfunctions for ergodic billiards.
\newblock 
 Comm.
Math. Phys. 175 (1996), 673$-$682.
\bibitem{EZB}
M. Zworski.
\newblock \emph{Semiclassical Analysis,} 
\newblock
Graduate Studies in Mathematics, AMS, 2012. 






%\bibitem{Sjos}
%J. Sj\"ostrand, \newblock 
%\emph{Singularites analytiques microlocales}.
%\newblock 
%!Asterisque, 95(1982).
\end{thebibliography}
\end{document}